\documentclass[11pt,letterpaper]{article}
  
\usepackage{amsmath,amsthm,nicefrac}

\usepackage{setspace}
\usepackage[breaklinks]{hyperref}
\hypersetup{colorlinks=true,
            citebordercolor={.6 .6 .6},linkbordercolor={.6 .6 .6},%
citecolor=blue,urlcolor=black,linkcolor=blue,pagecolor=black}
\usepackage[compact]{titlesec}
\usepackage[nameinlink]{cleveref}
\Crefname{algocf}{Algorithm}{Algorithms}
\crefname{algocfline}{line}{lines}
\Crefname{invariant}{Invariant}{Invariants}

\usepackage{epsfig}
\usepackage{amsthm,amssymb,mathrsfs}
\usepackage{xspace}
\usepackage{soul}
\usepackage{latexsym}
\usepackage{bbm}
\usepackage{enumitem}

\usepackage[dvipsnames]{xcolor}
\definecolor{DarkGray}{rgb}{0.66, 0.66, 0.66}
\definecolor{DarkPowderBlue}{rgb}{0.0, 0.2, 0.6}
\definecolor{fluorescentyellow}{rgb}{0.8, 1.0, 0.0}

\usepackage[ruled,vlined,linesnumbered,algonl]{algorithm2e}
\SetEndCharOfAlgoLine{}
\SetKwComment{Comment}{\footnotesize$\triangleright$\ }{}

\SetCommentSty{mycommfont}

\usepackage{thmtools,thm-restate}

\usepackage[letterpaper,top=2cm,bottom=2cm,left=3cm,right=3cm,marginparwidth=1.75cm]{geometry}
\makeatletter
\allowdisplaybreaks

\newcounter{note}[section]
\renewcommand{\thenote}{\thesection.\arabic{note}}
\ifdefined\DEBUG
\newcommand{\alert}[1]{{\color{red}#1}}
\newcommand{\agnote}[1]{\refstepcounter{note}$\ll${\bf Anupam~\thenote:}
  {\sf \color{blue} #1}$\gg$\marginpar{\tiny\bf AG~\thenote}}
\newcommand{\elnote}[1]{\refstepcounter{note}$\ll${\bf E~\thenote:}
  {\sf \color{gray} #1}$\gg$\marginpar{\tiny\bf EL~\thenote}}
\else
\newcommand{\alert}[1]{}
\newcommand{\agnote}[1]{}
\newcommand{\elnote}[1]{}
 \fi

\newcommand{\initOneLiners}{%
    \setlength{\itemsep}{0pt}
    \setlength{\parsep }{0pt}
    \setlength{\topsep }{0pt}
}

\newtheorem{theorem}{Theorem}[section]

\newcommand{\eat}[1]{}

\newcommand{\eps}{\varepsilon}
\newcommand{\sse}{\subseteq}

\newcommand{\R}{\mathbb{R}}

\newcommand{\N}{\mathbb{N}}

\newcommand{\cC}{\mathcal{C}}
\newcommand{\calC}{\mathcal{C}}
\newcommand{\calN}{\mathcal{N}}
\newcommand{\cF}{\mathcal{F}}
\newcommand{\calF}{\mathcal{F}}
\newcommand{\calI}{\mathcal{I}}
\newcommand{\cS}{\mathcal{S}}
\newcommand{\calS}{\mathcal{S}}
\newcommand{\cSdown}{\mathcal{S}^\downarrow}
\newcommand{\cN}{\mathcal{N}}

\newcommand{\poly}{\operatorname{poly}}

\renewcommand{\emptyset}{\varnothing}

\newcommand{\nf}{\nicefrac}

\newcommand{\width}{p}

\begin{document}

\title{A Local Search-Based Approach for Set Covering}

\author{Anupam Gupta\thanks{Carnegie Mellon University, Pittsburgh PA
    15217. Supported in part by NSF awards CCF-1955785, CCF-2006953, and CCF-2224718.} \and Euiwoong Lee\thanks{University of Michigan, Ann Arbor
    MI 48105. Partially supported by Google.} \and Jason Li\thanks{Simons Institute for the Theory of Computing and UC Berkeley, Berkeley CA 94720.}}
    
 \date{}
 
\maketitle

\begin{abstract}
  In the Set Cover problem, we are given a set system with each set
  having a weight, and we want to find a collection of sets that cover
  the universe, whilst having low total weight. There are several
  approaches known (based on greedy approaches, relax-and-round, and
  dual-fitting) that achieve a $H_k \approx \ln k + O(1)$
  approximation for this problem, where the size of each set is
  bounded by $k$. Moreover, getting a $\ln k - O(\ln \ln k)$
  approximation is hard. 

  Where does the truth lie? Can we close the gap between the upper
  and lower bounds? An improvement would be particularly interesting
  for small values of $k$, which are often used in reductions between
  Set Cover and other combinatorial optimization problems.

  We consider a non-oblivious local-search approach: to the best of
  our knowledge this gives the first $H_k$-approximation for Set Cover
  using an approach based on local-search. Our proof fits in one page,
  and gives a integrality gap result as well. Refining our approach by
  considering larger moves and an optimized potential function gives
  an $(H_k - \Omega(\log^2 k)/k)$-approximation, improving on the
  previous bound of $(H_k - \Omega(\nicefrac{1}{k^8}))$ (\emph{R.\
    Hassin and A.\ Levin, SICOMP '05}) based on a modified greedy
  algorithm.
\end{abstract}

\thispagestyle{empty}
\newpage{}
  
\setcounter{page}{1}

\section{Introduction}
\label{sec:introduction}

(Weighted) Set Cover is one of the most important problems in the approximation algorithms literature. 
Given a set system $(U, \calS)$ where each set $S \in \calS$ has
weight $w(S) > 0$, the Set Cover problem asks to find a subcollection
$\calF \subseteq \calS$ that covers the universe (i.e., $\cup_{S \in
  \calF} S = U$) while minimizing the total weight $\sum_{S \in \calF}
w(S)$. This problem is NP-hard as long as the sets have size at least
three (the edge-cover problem can be solved in polynomial time). 
As the flagship problem in two standard textbooks in approximation algorithms~\cite{Vazirani01, WS11}, 
and as an abstract setting capturing numerous covering problems, 
it has always been an important testbed for new algorithmic techniques. 

Let $k$-Set Cover be the special case of Set Cover where every set has
at most $k$ elements.
The simple greedy algorithm that iteratively selects the set
maximizing the current density (i.e., the ratio of the number of
uncovered elements in the set to its weight) guarantees an $H_k$-approximation, 
where $H_k = 1 + \nf12 +  \dots + \nf1k = \ln k +O(1)$ is the $k$th
harmonic number~\cite{johnson1974approximation, lovasz1975ratio,
  chvatal1979greedy}. It can be analyzed by the dual-fitting method,
upper bounding the integrality gap of the standard LP relaxation by
$H_k$ as well. Another proof of the integrality gap comes via the
relax-and-round approach (\cite{Young}, see also \S\ref{sec:r-and-r}). 
These algorithms are almost optimal due to the $(1-o(1)) \ln
n$-hardness of Feige~\cite{Feige98} and its refinement to the $\ln k -
O(\ln \ln k)$-hardness for $k$-Set Cover~\cite{Trevisan01}, which even
holds against $n^{f(k)}$-time algorithms for any computable function $f$. 

How about local search, one of the most intuitive and popular
algorithm design techniques? It maintains a solution (a set cover in
this case), and in each iteration, it tries to find a {\em local move}
that swaps at most $p$ sets between the current solution and the
remaining sets. If there exists a local move that results in a better
set cover, execute the local move; otherwise, output the current set
system. While it has been successfully applied for many problems
including bounded-degree network design~\cite{furer1992approximating,
  furer1994approximating}, Facility
Location~\cite{korupolu2000analysis, charikar2005improved} and
$k$-Median~\cite{arya2004local, gupta2008simpler}, the local search
cannot yield any finite approximation ratio for Set Cover, at least
when the local width is one; simply consider the example where the
universe has $k$ elements, the optimal solution contains $k$ singleton
sets of weight $\eps$, but the current solution consists of the set
containing all the elements but has large weight $1$. As $\eps \to 0$
for fixed $k$, the gap between the two solutions becomes unbounded.
 
Is there a way to ``redeem'' local search?  One reason that the above
example is bad for local search is that the {\em potential} that the
standard local search is trying to optimize, which is the same as the
total weight of the current solution, is too rigid; while adding a
singleton set from the optimal solution can be seen as a {\em
  progress}, since the large set is still needed to cover all
elements, adding the singleton set only worsens the total weight and
is not executed.  To fix this issue, {\em non-oblivious local search
  (NOLS)} tries to find a local move improving a carefully designed
potential different from the objective function of the problem.  
Originally defined by Khanna, Motwani, Sudan, and Vazirani~\cite{khanna1998syntactic}, 
it has been recently shown to work well for problems
including Submodular Optimization~\cite{filmus2014monotone}, Tree
Augmentation and Steiner Tree~\cite{TZ22} (which directly inspired
this paper), Steiner forest~\cite{GGKMSSV18} and
$k$-Median~\cite{cohen2022improved}.  Our first result is the
following ``redemption'' of local search for Set Cover showing that
NOLS with a natural 
potential can exactly match the $H_k$-approximation guarantee,
up to an arbitrarily small constant error $\eps > 0$ that ensures that
the local search terminates in polynomial time. 

\begin{theorem}
  \label{thm:main-one}
  For any $\eps > 0$, there exists a width-$1$ non-oblivious local search algorithm that can be implemented in time $\poly(n, 1/\eps)$ and yields an $H_k + \eps$-approximation for $k$-Set Cover. 
\end{theorem}

Our proof also shows that any local optimum has weight at most $H_k$
times an optimal solution to the LP relaxation for set cover, thereby
giving yet another proof of its integrality gap.

We then explore the power of NOLS beyond the $H_k$-approximation.
While Trevisan's $(\ln k - O(\ln \ln k))$-hardness shows that we
cannot dramatically improve it, there are still unanswered
questions, especially for small values of $k$, which are important for
hardness of other well-known problems including Steiner
Tree~\cite{bern1989steiner, thimm2001approximability}. 
For \emph{unweighted} $k$-Set Cover, there is a long series
of works~\cite{goldschmidt1993modified, halldorsson1995approximating,
  Halldorsson96, duh1997approximation, levin2009approximating,
  athanassopoulos2009analysis, FY11} giving an
$(H_k - \beta_k)$-approximation where $\beta_k \geq 0.5$ for every
$k \geq 3$ and approaches to $0.6402$. (Some of these works even use a
combination of oblivious local search to solve a packing problem, and
greedy to extend the solution to a matching.)

But the status for weighted $k$-Set Cover---which is the problem we
focus on---is understood far poorly. The best approximation
ratio remains $H_k - \Omega(\nicefrac{1}{k^8})$~\cite{HL05}, obtained
by a variant of the greedy algorithm. We make progress on this
direction, and prove the following improved approximation guarantees
for $k$-Set Cover.

\begin{theorem}
\label{thm:main-two-k}
For any $\eps > 0$, there exist width-$2$ and width-$k$ non-oblivious local search algorithms for $k$-Set Cover that yield $(H_k - \Omega(\nf{1}{k}) + \eps)$ and $(H_k - \Omega(\nf{(\log k)^2}{k}) + \eps)$-approximations respectively.
\end{theorem}

Note that a width-$2$ local search can be implemented in time
$\poly(n, 1/\eps)$ and a width-$k$ one can be na\"{\i}vely implemented in time
$n^{O(k)} \poly(1/\eps)$. 
We present clean locality gap results for width $1$, $2$, and $k$ swaps in Section~\ref{sec:set-cover},~\ref{sec:improvement-SC},~\ref{sec:moves-width-k} respectively and show how to implement them in polynomial time in Section~\ref{sec:poly-time-impl}.
In Section~\ref{sec:tight-lower-bounds}, we provide matching lower bounds showing that these two results are tight for a large class of natural potentials.

\section{Set Cover}
\label{sec:set-cover}

Consider a weighted set system $(U, \calS)$ with weights
$w : \calS \to \R^+$ where each $S \in \calS$ has cardinality at most
$k$.  Define the downwards closure $\cSdown$ of the set
system as containing all sets
$\{ T \mid \exists S \in \cS, T \sse S\}$, where each subset has the same
cost as the original set.  
Letting $\calS \leftarrow \cSdown$ does not change the optimal value, 
so we assume that $\calS$ is downwards-closed.
With this assumption, we can further assume the optimal solution $\cF^*$ 
forms a partition of the universe $U$, and our algorithms will maintain 
the solution $\cF$ that also forms a partition of $U$. 
(Letting $\calS \leftarrow \cSdown$ might significantly increase the number of sets. 
In Section~\ref{sec:poly-time-impl}, we show how to efficiently implement it.)

For any collection $\cF$ of sets that partition $U$, define the
\emph{Rosenthal potential}~\cite{rosenthal1973class}:
\begin{gather}
  \Phi(\cF) := \sum_{S \in \cF} w(S) H_{|S|}.   \label{eq:rosenthal}
\end{gather}
For each element $e$, let $\bar{w}_\mathcal{F}(e) := \frac{w(S)}{|S|}$ for the set
$S \in \cF$ that covers $e$. (We omit the subscript $\mathcal{F}$ if it is clear from context.)
Then $\sum_e \bar{w}(e) = w(\cF)$.

\subsection{An $H_k$-competitive Local Search Algorithm}
\label{sec:simpler}

Consider the following \emph{single (set) local moves}:
\begin{quote}
  Add in a single set $S \in \cSdown$, and then for each $T \in \cF$
  in the current solution, replace $T$ by $T \setminus S$ to get
  back a new set cover solution that is a partition.
\end{quote}
If this move decreases the Rosenthal potential---i.e., if this is an
\emph{improving} local move---we move to this resulting solution. 

\begin{theorem}[Single-Set Moves]
  Suppose $\cF$ is a local optimum, i.e., there are no improving local
  moves. Then $w(\cF) \leq H_k \cdot w(\cF^*)$. 
\end{theorem}

\begin{proof}
  To show the locality gap, we consider a specific set of local moves
  (called \emph{test moves}). Since there are no improving local
  moves, each of these test moves do not reduce the potential, thereby
  giving us relationships between the costs of some solutions related
  to the local and optimal solution. Combining these then proves the
  theorem.

  Indeed, consider using any of the sets in the optimal solution
  $S \in \cF^*$ as a local move from $\cF$. The resulting potential
  function change is
  \[ w(S) H_{|S|} - \sum_{T \in \cF} w(T) \big[ H_{|T|} - H_{|T
      \setminus S|} \big] \geq 0. \] 
      Since $\cF$ is a partition of $U$, the second term on the LHS is
  \begin{gather}
    \sum_{T \in \cF} \sum_{i = 0}^{|T\cap S| - 1} \frac{w(T)}{|T|-i}
   \geq \sum_{T \in \cF} \frac{w(T)}{|T|} |T \cap S| = \sum_{T \in
      \cF} \sum_{e \in S \cap T} \bar{w}(e) = \sum_{e \in S}
    \bar{w}(e). \label{eq:2}
  \end{gather}
    Therefore we have for each $S \in \cF^*$ that
  \begin{gather}
    w(S) H_{k} - \sum_{e \in S} \bar{w}(e) \geq 0. \label{eq:1}
  \end{gather}
  Summing over all sets in $\cF^*$, which we also imagine is a
  partition, we get
  \[ H_k \sum_{S \in \cF^*} w(S)  - \sum_e \bar{w}(e) \geq 0 \qquad
    \implies \qquad w(\cF) \leq H_{k} w(\cF^*) . \qedhere \]
\end{proof}

\subsection{An Integrality Gap Result}

A small change bounds the cost against any solution to the standard
linear programing relaxation:
\[ \min \big\{ \sum_S w(S) x_S \mid \sum_{S: e \in S} x_S \geq 1, x \geq 0
  \big\}. \] Indeed, suppose $x^*$ is any feasible solution, and we
consider local moves with each of the sets sets in the support of
$x^*$. Multiplying~(\ref{eq:1}) with $x^*_S$ and summing gives
\begin{gather}
  \sum_S H_{|S|} w(S) x^*_S - \sum_{e} \bar{w}(e) \sum_{S: e \in S}
  x^*_S \geq 0.
\end{gather}
But $\sum_{S: e \in S} x^*_S \geq 1$ by feasibility of the LP, so we
infer that
\[ \sum_e \bar{w}(e) =  w(\cF) \leq \sum_{S \in \cF^*} w(S)x^*_S \; H_{|S|} \leq H_k \cdot
  (w^\intercal x^*). \]

\section{An Improvement Using Double Moves}
\label{sec:improvement-SC}

The above analysis suggests one avenue for improvement: if the move
adding set $S \in \cF^*$ removes more than one element from some set
$T \in \cF$, then the inequality~(\ref{eq:2}) bounds the decrease in
potential by $\frac{|T \cap S|}{|T|}$, whereas the actual decrease is
$H_{|T|} - H_{|T \setminus S|}$, which is possibly greater. Concretely, if $|T| = k$ and
$|T \cap S| = 2$, then we claim an improvement of $\nf2k$, whereas the
actual improvement is $\nf1k + \nf1{k-1}$. In this section we show how
this idea can be used to get an improvement.

The algorithm is now a natural ``width-two'' generalization of the
above local search:
\begin{quote}
  Add in two sets $S, S' \in \cSdown$ to $\cF$, and replace each
  existing set $T \in \cF$ by $T \setminus (S \cup S')$. If the
  resulting partition has a smaller potential value, move to it.
\end{quote}
We allow $S=S'$, which captures the case of adding a single set. Hence
local optima with these moves have cost at most $H_k w(\cF^*)$ by the
previous section; we want to show a better bound. Let us first do this
for a special case, and then show how to remove this assumption
(in~\Cref{lem:postprocess}). 

\begin{theorem}[Double Moves]
  \label{thm:double}
  Consider a solution $\cF$ that is a local optimum for the above
  width-two local search with the Rosenthal potential. Let
  $\cF_1 \sse \cF$ be the subcollection of sets in $\cF$ having unit
  size, and suppose $w(\cF_1) \leq 0.99\, w(\cF)$. Then
  \[ w(\cF) \leq H_k(1- \Theta(\nf1{k^2}))\cdot w(\cF^*). \]
\end{theorem}

\begin{proof}
  The proof again goes via analyzing a collection of test moves; 
  these try to add in at most two sets at a time. To get the
  \emph{test moves}, consider a bipartite graph whose nodes are
  the sets in $\cF^*$ and those in $\cF$, and there is an edge between $S \in \cF^*$ and $T \in \cF$ iff $S\cap T\ne\emptyset$. 
  For a vertex $S$, let $\cN_S$ be the set of its neighbors.
  There are two kinds of test moves:
  \begin{enumerate}
  \item For a set $T \in \cF$, let
    $S_0,S_1,\ldots,S_{\ell-1}\in\mathcal F^*$ be its
    neighbors in an arbitrary order. If $\ell=1$, then try to add in
    $S_0$ twice. Else, for each index $0\le i<\ell$, try to add in
    $S_i$ and $S_{(i+1)\bmod\ell}$ together.
  \item A set $S\in \cF^*$ is added exactly $2|\mathcal N_S|$ times
    above, twice for each $T\in\mathcal N_S$. Add in $S$ another
    $2k - 2|\mathcal N_S|$ number of times.
  \end{enumerate}

  The local optimality ensures that none of the moves above decreases
  the potential.  Let us consider the total potential change caused by
  the above moves. Firstly, each set $S \in \cF^*$ is added exactly
  $2k$ times, so the total potential increase by adding it is exactly
  $2kw(S) H_{|S|} \leq 2k H_k \cdot w(S)$. Moreover, let us consider the potential decrease
  due to the removal of elements from each set in $\cF$. Indeed, for a
  set $T \in \cF$, let $\ell$ denote its neighborhood size in the
  bipartite graph. 
  \begin{itemize}
  \item If $\ell=1$, then $T$ is a subset of $S_0$, its only neighbor. 
  Thus the potential from $T$ is decreased by
    $w(T)H_{|T|}$ \emph{twice}, for a total of $2w(T)H_{|T|}$. If $\ell>1$,
    then for each pair $S,S'$ added together, the potential from $T$
    is decreased by $w(T)(H_{|T|}-H_{|T|-|(S \cup S') \cap
      T|})$. Since $|(S \cup S') \cap T| = |S\cap T|+|S'\cap T|\ge2$, the average per-element
    decrease,
    \[ \frac { w(T)(H_{|T|}-H_{|T|-|(S\cup S')\cap T|}) } {
        |(S \cup S')\cap T| } ,\] is at least what it would be if
    $|(S \cup S') \cap T|=2$, i.e., $w(T)(H_{|T|}-H_{|T|-2})/2$. So
    the overall decrease is at least
    \[ (|(S\cup S')\cap T|)\;w(T)\;(H_{|T|}-H_{|T|-2})/2 .\] Since the
    sum of $(|(S \cup S')\cap T|)$ over all added pairs $S,S'$ is
    exactly $2|T|$, the overall potential decrease is at least
    $|T|w(T)(H_{|T|}-H_{|T|-2})$.
    
    If $|T|=1$, then $\ell=1$ and the potential decrease is
    $2w(T)H_{|T|}=2w(T)$. If $|T|\ge2$, then regardless of whether
    $\ell=1$ or $\ell>1$, the potential decrease is at least
    $|T|w(T)(H_{|T|}-H_{|T|-2})$.
  \item In total, each $S \in \calN_T$ is added exactly $2k$ times,
    so each element of $T$ is removed a total of $2k$ times. Two of
    these $2k$ removals are accounted above, so the remaining
    potential decrease over all elements is at least
    $w(T) \cdot (2k-2) \cdot |T| \cdot (H_{|T|} -
    H_{|T|-1})=(2k-2)w(T) $.
  \end{itemize}
  The total potential decrease, which is at most $0$, is at least
  \begin{gather}
    \sum_{T\in\mathcal F_1}2w(T) + \sum_{T\in\mathcal
      F\setminus\mathcal F_1}|T|w(T)(H_{|T|}-H_{|T|-2})
    +\sum_{T\in\mathcal F} (2k-2)w(T)- H_k\sum_{S\in\mathcal
      F^*}2kw(S) , \label{eq:main-two}
  \end{gather}
    where $\cF_1$ is the collection of sets in
  $\cF$ of unit size. Observe that
  \[ |T|(H_{|T|}-H_{|T|-2}) =
    |T|\left(\frac1{|T|}+\frac1{|T|-1}\right)=2+\frac1{|T|-1} \ge
    2+\frac1{k-1}. \] By assumption,
  $w(\mathcal F\setminus\mathcal F_1)\ge 0.01\, w(\mathcal F)$, so the
  potential decrease is at least
  \begin{align*}
    &\sum_{T\in\mathcal F}\left(2+\frac1{100(k-1)}\right) w(T)
      +\sum_{T\in\mathcal F} (2k-2)w(T) - H_k \sum_{S\in\mathcal F^*}2kw(S)
    \\&= \sum_{T\in\mathcal F}\left( 2k+\frac1{100(k-1)}\right) w(T) - H_k\sum_{S\in\mathcal F^*}2kw(S).
  \end{align*}
  Since this decrease is at most $0$, we conclude that
  \[ w(\mathcal F)\le\frac{2k}{2k+\nicefrac1{(100(k-1))}}H_k \cdot w(\mathcal
  F^*) \le (1-\Theta(\nicefrac1{k^2}))\, H_k \cdot w(\mathcal F^*). \qedhere\]
\end{proof}

To get a better-than-$H_k$ approximation for all instances, we can go
two ways: the first approach is to post-process the locally
optimal solution $\cF$ using the following lemma (which we prove in
\S\ref{sec:post-proc-algor}) to handle the case not handled by~\Cref{thm:double}:
\begin{restatable}[Post-processing]{lemma}{PostProc}
  \label{lem:postprocess}
  Given any solution $\cF$ with $w(\cF) \leq H_k w(\cF^*)$, let
  $\cF_1$ be the subcollection of sets in $\cF$ of unit size. If
  $w(\cF_1) \geq 0.99\, w(\cF)$, there is an efficient algorithm that
  returns a new solution $\cF'$ with
  $w(\cF') \leq 0.99\, H_k \cdot w(\cF)$.
\end{restatable}
Hence,
returning the better of the solutions $\cF$ produced by the
local-search procedure, and $\cF'$ from using \Cref{lem:postprocess}
applied to $\cF$, gives a solution of cost at most
\[ H_k \cdot (1 - \Theta(\nf1{k^2}))\cdot w(\cF^*). \]
The second---better and more principled approach---is to modify the
potential function, which we do in the next section.

\subsection{An Improved Analysis using a Custom Potential Function}
\label{sec:improved-two-swap}

Let us consider a somewhat generic potential function: define
$f_1 :=1$, and let $f_i \geq 0$ be values to be fixed later,
satisfying $f_i \geq f_{i+1}$ for all $i$.  Let
$F_i:=\sum_{j=1}^if_j$, and define the following potential
\[ \Psi(\mathcal F)=\sum_{S\in\mathcal F}w(S)F_{|S|} .\] We get back
the Rosenthal potential by setting $f_i = 1/i$, but now we can
optimize over settings of $f_i$ to give better results. We again
consider the two-set local search algorithm, trying to reduce the
value of the new potential $\Psi(\mathcal F)$. The test moves remain
unchanged.

Moreover, the calculations remain essentially unchanged beyond replacing
$H_i$ by $F_i$; the argument about the average per-element decrease
being largest for $|(S \cup S') \cap T| = 2$ follows from the $f_i$
values being non-increasing. Consequently, the total
potential decrease, which is at most $0$, is at least
\begin{gather}
 {\sum_{T\in\mathcal F_1}2w(T) + \sum_{T\in\mathcal F\setminus\mathcal
    F_1}|T|w(T)(F_{|T|}-F_{|T|-2}) +\sum_{T\in\mathcal F}
  (2k-2)|T|w(T)f_{|T|} - F_k \sum_{S\in\mathcal F^*}2k\, w(S) .} \label{eq:main-two-prime}
\end{gather}
This equation can be compared
to~(\ref{eq:main-two}), where we had used the fact that the Rosenthal
potential satisfies $if_i = i(1/i) = 1$ and simplified the third summation above: 
$\sum_{T\in\mathcal F} (2k-2)|T|w(T)f_{|T|} = \sum_{T\in\mathcal F}
(2k-2)w(T)$. Let us abstract~(\ref{eq:main-two-prime}) further: define
\[ \alpha_t := \frac{1}{w(\cF)} \cdot \sum_{T \in \cF: |T|=t} w(T), \] and note
that $\alpha = (\alpha_1, \ldots, \alpha_k)$ gives a probability
distribution over the set sizes, and hence belongs to the probability
simplex $\triangle_k$. Dividing~(\ref{eq:main-two-prime}) through by
$2k$, simplifying  slightly, and using that this decrease is at most zero gives
\[ w(\cF) \cdot \bigg( \alpha_1 + \sum_{t \geq 2}
  \underbrace{\frac{t(F_t - F_{t-2}) + (2k-2) tf_t}{2k}}_{=: \phi_t}
  \, \alpha_t \bigg) \leq F_k\, w(\cF^*). \] Let $\phi_1 := 1$ and
$\phi_t$ be the coefficient of $\alpha_t$ as shown above.  Getting the
best approximation becomes an optimization problem: we want to set
$f_i$ values to minimize
$\max_{\alpha \in \triangle_k} \big\{ \frac{F_k}{\sum_t \phi_t
  \alpha_t} \big\}$, or equivalently, to minimize
$ F_k \cdot \max_t \big\{ \nicefrac{1}{\phi_t} \big\}$. Recall that we require $f_1=1$; we set
\begin{gather}
  f_t :=\frac1t-\frac1{4kt(t-1)} \text{ for $t>1$}. \label{eq:two-potential-choices}
\end{gather}
Then for $t\ge3$,
\begin{align*}
 t(F_{t}-F_{t-2}) &= t\left(\frac1{t}-\frac1{4kt(t-1)}+\frac1{t-1}-\frac1{4k(t-1)(t-2)}\right)
 \\&\ge 2+\frac1{t-1}-\frac1{2k(t-1)} \ge 2+\frac1{2(t-1)} , 
\end{align*}
and the bound $t(F_{t}-F_{t-2})\ge 2+\frac1{2(t-1)}$ can be separately verified for $t=2$.
Observe that $\phi_1 = 1$ by definition, and for $t > 1$, we have 
\[ \phi_t = \frac{t(F_t - F_{t-2}) + (2k-2) tf_t}{2k} \geq \frac{1}{2k}
  \bigg( 2 + \frac{1}{2(t-1)} + (2k-2) \bigg( 1 - \frac{1}{4k(t-1)}
  \bigg)\bigg) \geq 1.\]

Hence, the
approximation guarantee is at most $F_k \max_t (1/\phi_t) = F_k$.
Finally, we bound $F_k\le H_k-\nicefrac1{8k}$
since for $i=2$ alone, $f_2$ beats the corresponding term
$\nicefrac12$ from $H_k$ by $\nicefrac1{8k}$. 

\begin{theorem}[Two-Sets Moves] 
  \label{thm:two-swaps-better}
  Any local optimum for the two-sets local search using the potential
  $\Psi$ using the $f_i$ values from~(\ref{eq:two-potential-choices}) satisfies
  $w(\cF) \leq (H_k - \nf{1}{(8k)}) \cdot w(\cF^*)$. 
\end{theorem}

\section{Further Improvements: Moves with Width $k$}
\label{sec:moves-width-k}

We now consider the ``width-$k$'' generalization: add in sets $S_1,S_2,\ldots,S_k\in\mathcal S^\downarrow$ to $\mathcal F$, and replace each existing set $T \in \cF$ by $T \setminus (S_1 \cup \cdots \cup S_k)$. If the
  resulting partition has a smaller potential value, move to it. (Once again, we allow repeats in the sets, or equivalently, we allow moves of fewer than $k$ sets.)

Intuitively, if we choose sets $S_1,S_2,\ldots,S_k$ that cover a set $T\in\cF$ of size $k$, then $T$ disappears from $\cF$ and the improvement to the Rosenthal potential is $1+\nf12+\cdots+\nf1k$, or an average of $\nf1k\cdot(1+\nf12+\cdots+\nf1k)$, which is even better than the average $\nf12\cdot(\nf1{k-1}+\nf1k)$ in the width-$2$ case. We will actually use a custom potential function as before, but the Rosenthal potential provides a good baseline intuition.

We now define the test moves. For $T \in \cF$, let $\calN_T := \{ S \in \cF^* : |S \cap T| \neq 0 \}$ and similarly, for $S \in \cF^*$, let $\calN_S := \{ T \in \cF : |S \cap T| \neq 0 \}$.

\begin{enumerate}
\item For each $T \in \cF$, add the sets in $\calN_T$ together.

\item For each $S \in \cF^*$, add $S$. This move is multiplied by $(k - |\calN_S|)$ times so that each $S$ participates in exactly $k$ moves. 
\end{enumerate}

The local optimality ensures that none of the moves above decreases the potential. 
Let us consider the total potential change caused by the above moves.
First, each $S \in \cF^*$ is added exactly $k$ times, so the total potential increase by adding it is exactly $k w(S) F_{|S|}$, where we define the custom potential function $f_i$ later.

For $T \in \cF$, we consider the two types of moves separately.
\begin{itemize}
\item For the move when $\calN_T$ is added together, $T$ is removed from $\cF$, so the potential from $T$ is decreased by $w(T)F_{|T|}$. 
\item Other than this move, each $S \in \calN_T$ is added exactly $k - 1$ times, so each element of $T$ is removed $k - 1$ times more. Therefore, the total potential decrease for such moves is at least $w(T) \cdot (k-1) \cdot |T| \cdot (F_{|T|} - F_{|T|-1})$.
\end{itemize}

Therefore, the total potential decrease, which is at most $0$, is at least
\[ \sum_{T\in\mathcal F}w(T)F_{|T|} +\sum_{T\in\mathcal F}
  (k-1)|T|w(T)f_{|T|} - F_k \sum_{S\in\mathcal F^*}k\, w(S) .\]
We now follow the recipe from \S\ref{sec:improved-two-swap}: dividing
by $k$ and defining the probability distribution $\alpha_t := \frac{1}{w(\cF)} \cdot \sum_{T \in \cF: |T|=t} w(T)$ gives
\[ w(\cF) \cdot \bigg(  \sum_{t \geq1}
  \underbrace{\frac{F_t + (k-1) tf_t}{k}}_{=: \phi_t}
  \, \alpha_t \bigg) \leq F_k\, w(\cF^*) .\]

Let $\phi_t$ be the coefficient of $\alpha_t$ as shown above.  Getting the
best approximation is again an optimization problem: we want to set
$f_i$ values to minimize
$\max_{\alpha \in \triangle_k} \big\{ \frac{F_k}{\sum_t \phi_t
  \alpha_t} \big\}$, or equivalently, to minimize
$ F_k \cdot \max_t \big\{ \nicefrac{1}{\phi_t} \big\}$. We set
\begin{gather}
  f_t :=\frac1t-\frac{\log t}{8kt} \text{ for $t\ge1$}. \label{eq:two-potential-choices}
\end{gather}
To verify that $\max_t\{\nf1{\phi_t}\}\le1$, we first bound $F_t$ by
\[ F_t = \sum_{i=1}^tf_i = H_t-\sum_{i=1}^t\frac{\log i}{8ki} \ge H_t-\int_{x=1}^{t+1}\frac{\log x}{8kx}dx=H_t-\frac{\log^2(t+1)}{16k} \ge H_t-\frac{\log(t+1)}{16} \]
using $\log(t+1)\le t\le k$ for the last inequality.
Therefore,
\begin{align*}
 k\phi_t=F_t+(k-1)tf_t&=H_t-\frac{\log(t+1)}{16}+(k-1)t\left(\frac1t-\frac{\log t}{8kt}\right)
\\&\ge(k-1)+H_t-\frac{\log(t+1)}{16}-\frac{\log t}8
\\&\ge k\text{ for }t\ge2.
\end{align*}
We can show $k\phi_1\ge k$ separately for $t=1$ using $F_1=f_1=1$. Therefore, the approximation ratio is at most $F_k = H_k - \Theta(\nf{(\log k)^2}{k})$.

\section{A Polynomial-Time Implementation}
\label{sec:poly-time-impl}

There are two issues with the running time of the above algorithms: (a)
since we consider adding sets from the exponentially-large
subset-closed family of sets (i.e., we assumed $\calS = \cSdown$), 
finding such a feasible local move may
not naively be polynomial-time implementable. Moreover, (b)~reaching a
local optimum may not be feasible in polynomial time. The second issue
can be handled using the standard technique of stopping when none of
the local moves decrease the potential by more than a $\delta w(\calF)$ 
(see, e.g.,~\cite[\S9.1]{WS11}). 
By setting $\delta = \eps / |U|$ and changing the RHS of~\eqref{eq:1}
from $0$ to $\delta w(\calF)$ and using $|\calF^*| \leq |U|$, 
we can ensure that $w(\calF) \leq \frac{H_k}{1 - \eps} w(F^*)$ when
there is no such improving move. 
Assuming the initial solution is $\poly(n)$-approximate, one can ensure that
the running time is $\poly(n, \eps)$. 


For issue~(a), let $\calS$ be the original collection of sets, not necessarily downwards closed. 
Suppose that we have the current solution $\calF$ and $S_1, \dots, S_p \in \calS$, 
and want to find appropriate pairwise disjoint subsets $S'_1 \subseteq S_1, \dots, S'_p \subseteq S_p$ such that the new solution that adds $S'_1, \dots, S'_p$ to $\calF$ (and subtracts their union from every $S \in \calF$) has a low potential. 
We do not know of an efficient way to compute the optimal choice of $S'_1, \dots, S'_p$. (One possible solution is, letting $T_1, \dots, T_q \in \calF$ be the sets intersecting $\cup_{i \in [p]} S_i$ (so $q \leq pk$), to (1) guess $|S'_i|$ for every $i \in [p]$, and $|T_j \cap (\cup_i S'_i)|$ for every $j \in [q]$ that exactly determine the potential change and (2) set up a network flow testing whether such $S'_1, \dots, S'_p$ exist, but it takes time $k^{O(pk)} \poly(n)$.) 

A more efficient implementation without necessarily finding the optimal $S'_1, \dots, S'_p$ is this: in all our previous proofs, when we considered adding $S \in \calF^*$ to the solution, the analysis always used $F_k w(S)$ as (an upper bound on) the increase of the potential by adding $S$ instead of the exact increase $F_{|S|} w(S)$. 
This means that we can indeed run more conservative local search; given $S_1, \dots, S_p$, 
let $A = \sum_{i \in [p]} w(S_i) F_k$ be the total increase from adding them, 
compute the total decrease $B$ caused by removing $\cup_{i \in [p]} S_i$ from the current sets in $\calF$, and only execute the local move when $A$ is smaller than $B$ (by $\delta w(\calF)$). 
All our analyses prove that we achieve the claimed approximation guarantees even when such a more conservative local move is not possible. Of course, if $S_1, \dots, S_p$ overlap, we can arbitrarily drop elements from them to ensure that the new solution is a partition as well; this further drops the potential.

\section{Tight Lower Bounds}
\label{sec:tight-lower-bounds}

In this section, we show that the approximation ratios of $H_k -
\Theta(\nf 1k)$ and $H_k - \Theta(\nf{(\log k)^2}{k})$ achieved by the
width-$2$ and width-$k$ local search respectively are optimal. In
fact, they are optimal under any potential of the form $ \Phi(\calF) = \sum_{S \in \calF} w(S) F_{|S|}$ where $F_{\ell} = f_1 + \dots + f_{\ell}$ for some $f_1 \geq \dots \geq f_k$. 
(We believe that the monotonicity condition is unnecessary, but currently do not have a formal proof.)

\subsection{Width-$2$ Lower Bounds}
\label{sec:width_two_lb}

We construct a collection of lower bound instances $\calI_{\ell}$, one for each $\ell
\in [k]$. The instance $\calI_{\ell}$ is the following:
\begin{enumerate}[nolistsep]
\item The universe is $U$. 
\item The optimum $\calF^*$ partitions $U$ into sets of size exactly $k$ where each set has cost $1$.
\item The local optimum $\calF$ partitions $U$ into sets of size
  exactly $\ell$ where each set has cost $\alpha_{\ell}$, to be
  determined below.
\item The set system is $(U, \calF \cup \calF^*)$. 
\item Let $G_{\ell}$ be a bipartite graph with $\calF$ and $\calF^*$
  as two sides, where each element $u \in U$ corresponds an edge
  connecting the two sets that contain $u$. We can ensure that the girth
  of $G_{\ell}$ is at least a constant arbitrarily larger than $k$~\cite{furedi1995graphs}.
\end{enumerate}

Let us determine the value of $\alpha_{\ell}$ so that $\calF$ becomes a local optimum. 
When $\ell = 1$, $\alpha_1 = \nf{F_k}{k}$ suffices, which yields the approximation ratio $F_k$.

For $\ell \geq 2$, there are essentially two kinds of moves: there are
width-$2$ local moves that add sets $S_1, S_2 \in \cF^*$ such that
$|(S_1 \cup S_2) \cap T| \leq 1$ for all $T \in \cF$, or to
add two distinct sets
$S_1, S_2 \in \calF^*$ that intersect a common $T \in \calF$. (The
girth condition ensures that there can be at most one such $T$, and $|S_1
\cap T| = |S_2 \cap T| = 1$. 
\begin{enumerate}[label=(\roman*)]
\item The potential increase due to adding $S_1, S_2$ is $2 F_k$.
\item Either each set $T \in \cF$ intersecting $S_1 \cup S_2$ loses
  one element, or some $T$ loses two elements and all the other sets
  intersecting $S_1$ or $S_2$ lose exactly one element. So the potential decrease due to removing elements from sets in $\calF$ is 
  \[
    (f_{\ell} + f_{\ell-1}) + (2k - 2) f_{\ell}.
  \]
  (Here we use the fact that $f_{\ell-1} \geq f_{\ell}$.)
\item Therefore, $\calF$ is a local optimum as long as 
  \[
    \alpha_{\ell} = \frac{2F_k}{(2k - 1)f_{\ell} + f_{\ell - 1}},
  \]
  and the approximation ratio in this case is 
  \[
    \frac{k \alpha_k}{\ell} = \frac{k}{\ell} \cdot \frac{2F_k}{(2k - 1)f_{\ell} + f_{\ell - 1}}
    = \frac{F_k}{\ell (f_{\ell} + \nf{f_{\ell - 1} - f_{\ell}}{2k} )}.
  \]
  \agnote{Not sure doing this is making things better, see comments
    below in red.}
\end{enumerate}
Fixing $f_1 = 1$ and optimizing $f_2, \dots, f_k$ to minimize the
worst-case approximation ratio over the $k$ instances $\calI_1, \dots, \calI_k$ shows that the best possible approximation ratio is determined by setting 
\[
f_{\ell} = \frac{1}{\ell} + \frac{1}{2k - 1} \bigg(\frac{1}{\ell} - f_{\ell - 1} \bigg)
\]
for $\ell = 2, \dots, k$. It yields $f_{\ell} = \nf{1}{\ell} - \Theta(\nf{1}{k\ell^2})$ as in the upper bound proof in Section~\ref{sec:improvement-SC}, showing that no potential can guarantee strictly better than $H_k - \Theta(\nf{1}{k})$. 

\alert{If I understand right, we're solving the LP:
  \begin{alignat*}{2}
    \max & ~~~~\lambda && \\
    f_1 + \ldots + f_k &= 1  && \\
    f_1 &\geq \lambda  && \\
    \ell f_\ell &\geq \lambda  & \qquad\qquad &\forall \ell \geq 2 \\
    \ell ( f_\ell + (f_{\ell-1}- f_\ell)/(2k) ) &\geq \lambda && \forall \ell \geq 2 \\
    f &\geq 0. &&
  \end{alignat*}
  Here $\lambda$ is the denominator above, so we're trying to maximize
  it. Now if we drop the constraints $\lambda \leq \ell f_\ell$ and
  solve the LP, and the solution satisfies the dropped constraints,
  then we're done.

  If so, could we just say: we'll find the optimal values of $f$ for
  just these ``worst-case'' moves. Then we show that for the weights
  we found, the other moves cannot do better, and hence we must be
  optimal? Is this argument kosher?

  BTW, should we fit a dual to show optimality? Or we just claim
  optimality?}
\elnote{Since we are showing that $f_{\ell} + \max(f_{\ell-1} - f_{\ell}, 0) / 2k \geq \lambda$, we cannot put two constraints separately?} 
\subsection{Width-$\width$ Lower Bounds}
\label{sec:width_w_lb}

\agnote{$w$ was the weight so changed to $\width$. It's a macro, so can
  change to whatever we like.}
The lower bound for the case of width-$\width$ follows the same framework.  Fix $\ell \in [k]$ and
consider the instance $\calI_{\ell}$ defined in
Section~\ref{sec:width_two_lb} (while ensuring that the girth
$\gg \width$), and determine the value of
$\alpha_{\ell}$ so that $\calF$ becomes a local optimum.  When
$\ell = 1$, $\alpha_1 = \nf{F_k}{k}$ suffices, which yields the
approximation ratio $F_k$.

For $\ell \geq 2$, let us consider what the best local width-$\width$ moves would be.
For $S_1, \dots, S_{\width}$ from the optimal solution $\calF^*$, consider the bipartite graph where the left vertices are $S_1, \dots, S_{\width}$, the right vertices are the sets from the current solution $\calF$ intersecting $S_1, \dots, S_{\width}$, and there is an edge if two sets intersect. Since the girth of the instance is much larger than $W$, this bipartite graph is a tree with exactly $Wk$ edges and $Wk + 1$ vertices, so the number of right vertices is $P = W(k - 1) + 1$. If we let $d_1 \geq \dots \geq d_P$ be the degrees of the right vertices, the potential decrease from the current set is 
\[
\sum_{i=1}^P (F_{\ell} - F_{\ell - d_i}). 
\]
Since both $P$ and $\sum_{i=1}^P d_i = Wk$ are fixed, the monotonicity of $f_1 \geq \dots \geq f_k$ implies that the above is when the degree is {\em maximally skewed}; 
defining $t \in \N$ and $1 \leq r < \ell - 1$ such that $\width - 1 = t(\ell - 1) + r$, we have
$t$ right vertices have degree $\ell$, one right vertex has degree $r + 1$, and the remaining $\width(k-1) - t$ right vertices have degree $1$. 
 As a sanity check, note that $t \cdot \ell + (r + 1) + \width(k - 1)
 - t = (t(\ell - 1) + r) + 1 + \width(k - 1) = \width k$.

With this move,
\begin{enumerate}[label=(\roman*)]
\item The potential increase due to adding $S_1, S_2, \dots S_\width$ is $\width F_k$.
\item The potential decrease due to removing elements from sets in $\calF$ is 
\[
t F_{\ell} + (F_{\ell} - F_{\ell - r - 1}) + (\width(k-1) - t)f_{\ell}.
\]
\item Therefore, $\calF$ becomes a local optimum if 
\[
\alpha_{\ell} = \frac{\width F_k}{t F_{\ell} + (F_{\ell} - F_{\ell - r - 1}) + (\width(k-1) - t)f_{\ell}},
\]
and the approximation ratio in this case is $\alpha_{\ell} \cdot (\nf k \ell)$.
\end{enumerate}

Again fixing $f_1 = 1$ and optimizing $f_2, \dots, f_k$ to minimize
the approximation ratio for $\calI_1, \dots, \calI_k$, yields
$f_{\ell} = \nf{1}{\ell} - \Theta(\nf{\log \ell}{k\ell})$ just like
we used in \S\ref{sec:moves-width-k} for the upper bound for $k$-moves.
Intuitively, setting $s = \nicefrac{(\width - 1)}{(\ell - 1)}$ so that $t(\ell - 1) + r = s(\ell - 1)$, the approximation ratio $\alpha \cdot (\nf k \ell)$ becomes
\begin{align*}
& 
\frac{k}{\ell} \cdot 
\frac{\width F_k}{t F_{\ell} + (F_{\ell} - F_{\ell - r - 1}) + (\width(k-1) - t)f_{\ell}}
\\
\geq & 
\frac{k}{\ell} \cdot 
\frac{\width F_k}{s F_{\ell} + (\width k - s\ell)f_{\ell}}
\\
= & 
\frac{F_k}{\ell \bigg( \frac{s F_{\ell}}{\width k} + (1- \nf{s\ell}{\width k})f_{\ell} \bigg)}
\\
= & 
\frac{F_k}{ \Theta(\frac{F_{\ell}}{k}) + (1- \Theta(\nf{1}{k}))\ell f_{\ell}},
\end{align*}
so that with $F_{\ell} = \Theta(\log \ell)$, the denominator becomes at least $1$ when $\ell f_{\ell} = (1 - \Theta(\nf {\log \ell}{k}))$. 
Therefore, no potential can guarantee strictly better than $H_k - \Theta(\nf{(\log k)^2}{k})$.

{\small
\bibliographystyle{alpha}
\bibliography{mybib}

\newcommand{\etalchar}[1]{$^{#1}$}
\begin{thebibliography}{CAGH{\etalchar{+}}22}

\bibitem[ACK09]{athanassopoulos2009analysis}
Stavros Athanassopoulos, Ioannis Caragiannis, and Christos Kaklamanis.
\newblock Analysis of approximation algorithms for k-set cover using
  factor-revealing linear programs.
\newblock {\em Theory of computing systems}, 45(3):555--576, 2009.

\bibitem[AGK{\etalchar{+}}04]{arya2004local}
Vijay Arya, Naveen Garg, Rohit Khandekar, Adam Meyerson, Kamesh Munagala, and
  Vinayaka Pandit.
\newblock Local search heuristics for k-median and facility location problems.
\newblock {\em SIAM Journal on computing}, 33(3):544--562, 2004.

\bibitem[BP89]{bern1989steiner}
Marshall Bern and Paul Plassmann.
\newblock The {Steiner} problem with edge lengths 1 and 2.
\newblock {\em Information Processing Letters}, 32(4):171--176, 1989.

\bibitem[CAGH{\etalchar{+}}22]{cohen2022improved}
Vincent Cohen-Addad, Anupam Gupta, Lunjia Hu, Hoon Oh, and David Saulpic.
\newblock An improved local search algorithm for k-median.
\newblock In {\em Proceedings of the 2022 Annual ACM-SIAM Symposium on Discrete
  Algorithms (SODA)}, pages 1556--1612. SIAM, 2022.

\bibitem[CG05]{charikar2005improved}
Moses Charikar and Sudipto Guha.
\newblock Improved combinatorial algorithms for facility location problems.
\newblock {\em SIAM Journal on Computing}, 34(4):803--824, 2005.

\bibitem[Chv79]{chvatal1979greedy}
Vasek Chvatal.
\newblock A greedy heuristic for the set-covering problem.
\newblock {\em Mathematics of operations research}, 4(3):233--235, 1979.

\bibitem[DF97]{duh1997approximation}
Rong-chii Duh and Martin F{\"u}rer.
\newblock Approximation of k-set cover by semi-local optimization.
\newblock In {\em Proceedings of the twenty-ninth annual ACM symposium on
  Theory of computing}, pages 256--264, 1997.

\bibitem[Fei98]{Feige98}
Uriel Feige.
\newblock A threshold of {$\ln n$} for approximating set cover.
\newblock {\em Journal of the ACM (JACM)}, 45(4):634--652, 1998.

\bibitem[FLS{\etalchar{+}}95]{furedi1995graphs}
Z~Furedi, Felix Lazebnik, A~Seress, Vasiliy~A Ustimenko, and Andrew~J Woldar.
\newblock Graphs of prescribed girth and bi-degree.
\newblock {\em Journal of Combinatorial Theory, Series B}, 64(2):228--239,
  1995.

\bibitem[FR92]{furer1992approximating}
Martin F{\"u}rer and Balaji Raghavachari.
\newblock Approximating the minimum degree spanning tree to within one from the
  optimal degree.
\newblock In {\em Proceedings of the third annual ACM-SIAM symposium on
  Discrete algorithms}, pages 317--324, 1992.

\bibitem[FR94]{furer1994approximating}
Martin Furer and Balaji Raghavachari.
\newblock Approximating the minimum-degree steiner tree to within one of
  optimal.
\newblock {\em Journal of Algorithms}, 17(3):409--423, 1994.

\bibitem[FW14]{filmus2014monotone}
Yuval Filmus and Justin Ward.
\newblock Monotone submodular maximization over a matroid via non-oblivious
  local search.
\newblock {\em SIAM Journal on Computing}, 43(2):514--542, 2014.

\bibitem[FY11]{FY11}
Martin F{\"u}rer and Huiwen Yu.
\newblock Packing-based approximation algorithm for the k-set cover problem.
\newblock In {\em International Symposium on Algorithms and Computation}, pages
  484--493. Springer, 2011.

\bibitem[GGK{\etalchar{+}}18]{GGKMSSV18}
Martin Gro{\ss}, Anupam Gupta, Amit Kumar, Jannik Matuschke, Daniel~R. Schmidt,
  Melanie Schmidt, and Jos{\'{e}} Verschae.
\newblock A local-search algorithm for steiner forest.
\newblock In {\em ITCS}, pages 31:1--31:17, Jan 2018.

\bibitem[GHY93]{goldschmidt1993modified}
Olivier Goldschmidt, Dorit~S Hochbaum, and Gang Yu.
\newblock A modified greedy heuristic for the set covering problem with
  improved worst case bound.
\newblock {\em Information processing letters}, 48(6):305--310, 1993.

\bibitem[GT08]{gupta2008simpler}
Anupam Gupta and Kanat Tangwongsan.
\newblock Simpler analyses of local search algorithms for facility location.
\newblock {\em arXiv preprint arXiv:0809.2554}, 2008.

\bibitem[Hal95]{halldorsson1995approximating}
Magn{\'u}s~M Halld{\'o}rsson.
\newblock Approximating discrete collections via local improvements.
\newblock In {\em SODA}, volume~95, pages 160--169, 1995.

\bibitem[Hal96]{Halldorsson96}
Magn{\'{u}}s~M. Halld{\'{o}}rsson.
\newblock Approximating \emph{k}-set cover and complementary graph coloring.
\newblock In William~H. Cunningham, S.~Thomas McCormick, and Maurice Queyranne,
  editors, {\em Integer Programming and Combinatorial Optimization, 5th
  International {IPCO} Conference, Vancouver, British Columbia, Canada, June
  3-5, 1996, Proceedings}, volume 1084 of {\em Lecture Notes in Computer
  Science}, pages 118--131. Springer, 1996.

\bibitem[HL05]{HL05}
Refael Hassin and Asaf Levin.
\newblock A better-than-greedy approximation algorithm for the minimum set
  cover problem.
\newblock {\em SIAM Journal on Computing}, 35(1):189--200, 2005.

\bibitem[Hoc82]{hochbaum1982approximation}
Dorit~S Hochbaum.
\newblock Approximation algorithms for the set covering and vertex cover
  problems.
\newblock {\em SIAM Journal on computing}, 11(3):555--556, 1982.

\bibitem[Joh74]{johnson1974approximation}
David~S Johnson.
\newblock Approximation algorithms for combinatorial problems.
\newblock {\em Journal of computer and system sciences}, 9(3):256--278, 1974.

\bibitem[KMSV98]{khanna1998syntactic}
Sanjeev Khanna, Rajeev Motwani, Madhu Sudan, and Umesh Vazirani.
\newblock On syntactic versus computational views of approximability.
\newblock {\em SIAM Journal on Computing}, 28(1):164--191, 1998.

\bibitem[KPR00]{korupolu2000analysis}
Madhukar~R Korupolu, C~Greg Plaxton, and Rajmohan Rajaraman.
\newblock Analysis of a local search heuristic for facility location problems.
\newblock {\em Journal of algorithms}, 37(1):146--188, 2000.

\bibitem[Lev09]{levin2009approximating}
Asaf Levin.
\newblock Approximating the unweighted k-set cover problem: greedy meets local
  search.
\newblock {\em SIAM Journal on Discrete Mathematics}, 23(1):251--264, 2009.

\bibitem[Lov75]{lovasz1975ratio}
L{\'a}szl{\'o} Lov{\'a}sz.
\newblock On the ratio of optimal integral and fractional covers.
\newblock {\em Discrete mathematics}, 13(4):383--390, 1975.

\bibitem[Ros73]{rosenthal1973class}
Robert~W Rosenthal.
\newblock A class of games possessing pure-strategy nash equilibria.
\newblock {\em International Journal of Game Theory}, 2(1):65--67, 1973.

\bibitem[Svi04]{Sviridenko04}
Maxim Sviridenko.
\newblock A note on maximizing a submodular set function subject to a knapsack
  constraint.
\newblock {\em Operations Research Letters}, 32(1):41--43, 2004.

\bibitem[Thi01]{thimm2001approximability}
Martin Thimm.
\newblock On the approximability of the {Steiner} tree problem.
\newblock In {\em International Symposium on Mathematical Foundations of
  Computer Science}, pages 678--689. Springer, 2001.

\bibitem[Tre01]{Trevisan01}
Luca Trevisan.
\newblock Non-approximability results for optimization problems on bounded
  degree instances.
\newblock In {\em Proceedings of the thirty-third annual ACM symposium on
  Theory of computing}, pages 453--461, 2001.

\bibitem[TZ22]{TZ22}
Vera Traub and Rico Zenklusen.
\newblock Local search for weighted tree augmentation and {Steiner} tree.
\newblock In {\em Proceedings of the 2022 Annual ACM-SIAM Symposium on Discrete
  Algorithms (SODA)}, pages 3253--3272. SIAM, 2022.

\bibitem[Vaz01]{Vazirani01}
Vijay~V Vazirani.
\newblock {\em Approximation algorithms}, volume~1.
\newblock Springer, 2001.

\bibitem[WS11]{WS11}
David~P Williamson and David~B Shmoys.
\newblock {\em The design of approximation algorithms}.
\newblock Cambridge university press, 2011.

\bibitem[You22]{Young}
Neal Young.
\newblock Greedy set cover {III}: weighted {H(d)}-approximation via localizing.
\newblock unpublished notes, 2022.
\newblock
  \href{https://algnotes.info/on/obliv/greedy/set-cover-local/}{https://algnotes.info/on/obliv/greedy/set-cover-local/}.

\end{thebibliography}
}

\appendix

\section{The Post-processing Algorithm}
\label{sec:post-proc-algor}

\PostProc*

\begin{proof}
  We assume that $k \geq 3$, else the resulting edge-cover problem
  can be solved exactly in polynomial time.
  Suppose $w(\cF) \leq 0.99\, H_k
  \cdot w(\cF^*)$, then we are already
  done, so assume otherwise. The assumption of the lemma means
  $w(\cF_1) \geq 0.98 \,H_k \cdot w(\cF^*)$. 

  Note that adding $S \in \cF^*$ immediately allows us to remove the
  singleton sets which are contained in $S$. For a collection of sets
  $\calC \subseteq \calS$, let
  $\calN_{\calC, 1} := \{ T \in \calF_1 \mid T \subseteq \cup_{S \in
    \calC} S \}$ be the collection of singletons from $\cF_1$ that can
  be removed from the solution by adding $\calC$.
  
  We set up an instance of {\em Knapsack Cover}: 
  we seek a collection $\calC \subseteq \calS$
  with $w(\calC) \le w(\cF^*)$ that maximizes the {\em saving}
  $w(\calN_{\calC, 1})$. Since $\cF^*$ is a feasible solution with
  saving at least $w(\cF_1) \geq 0.98\, H_k \cdot w(\cF^*)$, we
  can use an $(1-\nicefrac{1}{e})$-approximation
  algorithm~\cite{Sviridenko04} for knapsack cover to find a feasible
  solution $\calC$ having weight at most $w(\cF^*)$ and savings
  at least $(1-\nicefrac{1}{e}) \cdot 0.98 \, H_k \cdot
  w(\cF^*)$. This means $\calF \cup \calC \setminus \calN_{\calC, 1}$ is a set
  cover with cost at most
  \[ \bigg(1 + H_k(1 -
    0.98 \, (1-\nicefrac{1}{e}))\bigg)\, w(\cF^*) \leq 0.99\, H_k w(\cF^*), \]
  for all integers $k \geq 3$.
\end{proof}

\section{The $H_k$ Bound via Relax-and-Round}
\label{sec:r-and-r}

The traditional analysis of relax-and-round achieves a bound of $O(\ln
k)$~\cite{hochbaum1982approximation}, but one can tighten the bound to $H_k$. Consider the
following algorithm:
\begin{quote}
  Solve the LP relaxation to get solution $x^*$. Repeatedly pick sets
  $S_1, S_2, \ldots$ from $\cS$, each time picking set $S$ with
  probability $\frac{x^*_S}{\sum_T x^*_T}$. Finally, let $\cF$ be the
  sets $S_i$ that cover elements not covered by previous sets
  $S_1, \ldots, S_{i-1}$.
\end{quote}

The following claim is proved in lecture notes of Young~\cite{Young}:
\begin{theorem}
  The algorithm above incurs expected cost at most $H_k \cdot \sum_S w(S) x^*_S$.
\end{theorem}

\alert{I reconstructed a proof that Niv had shown me: but I need to
  check with him before giving this one:

\begin{proof}
  We bound the probability that $S$ is part of the solution.  Without
  loss of generality, assume that no set in the support of $x^*$
  contains another. We proceed by induction, and show that
  $\Pr[S \in \cF] \leq H_{|S|} \cdot x_S^*$. Let $\cC$ be the sets
  that intersect $S$ (including $S$), $\cC' := \cC \setminus \{S\}$,
  and let $Z := \sum_{T \in \cC} x_T^*$.  We can focus on only the
  sets in $\cC$, and using induction (on the size of the set) to argue
  that
  \begin{align*}
    \Pr[S \in \cF] &= \frac{x_S^*}{Z} + \sum_{T \in \cC'}
                     \frac{x_T^*}{Z} \cdot x_S^* H_{|S\setminus T|} 
                     = \frac{x_S^*}{Z} \cdot \bigg( 1 + \sum_{T \in \cC'}
                     x_T^* H_{|S\setminus T|} \bigg) \\
                   &\leq \frac{x_S^*}{Z} \cdot \bigg( \sum_{T \in \cC} x^*_T \frac{|T\cap S|}{|S|}  + \sum_{T \in \cC'}
                     x_T^* H_{|S\setminus T|} \bigg); \\
    \intertext{here we use that $\sum_{T \in \cC: e \in T} x^*_T \geq 1$ for
    each $e \in S$. Indeed, summing these inequalities over all $e \in S$ and reversing the order of
    summation gives $\sum_{T \in \cC} x^*_T \; |T \cap S| \geq
    |S|$, which is what we used above. However, $H_{|S \setminus T|} + \frac{|T\cap S|}{|S|} \leq
    H_{|S|}$, giving}
                   &\leq \frac{x_S^*}{Z} \cdot \bigg( x^*_S + \sum_{T \in \cC'}
                     x_T^* H_{|S|} \bigg) \leq \frac{x_S^*}{Z} \cdot
                     H_{|S|}\,Z = H_{|S|} \; x_S^* , 
  \end{align*}
  and hence the theorem.
\end{proof}
}
\end{document}